\documentclass[letterpaper,11pt]{article}
\usepackage[margin=1in]{geometry}
\usepackage[english]{babel}
\usepackage{amsmath, amsfonts, amssymb, amsthm}
\usepackage{thm-restate}
\usepackage{bbm,url}
\usepackage[colorlinks=true,urlcolor=black,linkcolor=black,citecolor=black]{hyperref}
\usepackage{graphicx}
\usepackage{color}
\usepackage{subcaption}
\usepackage{hhline}
\usepackage{pifont}

\usepackage{nicefrac,bm}
\usepackage{algorithm}
\usepackage[noend]{algpseudocode}
\usepackage{thmtools}
\usepackage{thm-restate}

\allowdisplaybreaks

\newtheorem{theorem}{Theorem}
\newtheorem*{theorem*}{Theorem}

\newtheorem{lemma}{Lemma}

\newtheorem{remark}[theorem]{Remark}

\newtheorem{example}[theorem]{Example}

\theoremstyle{definition}
\newtheorem{definition}{Definition}

\newcommand{\hide}[1]{}

\newcommand{\s}[1]{\mathsf{#1}}
\newcommand{\NW}{\s{NW}}

\newcommand{\G}{\mathcal{G}}

\newcommand{\x}{\mathbf{x}}
\newcommand{\y}{\mathbf{y}}
\newcommand{\z}{\mathbf{z}}
\newcommand{\p}{\mathbf{p}}
\newcommand{\q}{\mathbf{q}}

\newcommand{\N}{\mathbb N}
\newcommand{\Z}{\mathbb Z}

\newcommand{\R}{\mathbb R}
\newcommand{\poly}[1]{\mathsf{poly}(#1)}
\newcommand{\MBB}{\text{MBB}}
\def\eps{\varepsilon}

\title{Computing Pareto-Optimal and Almost Envy-Free \\ Allocations of Indivisible Goods\thanks{A preliminary version of the paper appeared at AAAI 2021 \cite{GargM21}. Work supported by the NSF Grant CCF-1942321 }}

\author{Jugal Garg\footnote{University of Illinois at Urbana-Champaign, USA} \\
\texttt{\small jugal@illinois.edu} 
\and
Aniket Murhekar\footnote{University of Illinois at Urbana-Champaign, USA}\\
\texttt{\small aniket2@illinois.edu}
}

\date{}
\begin{document}

\maketitle

\begin{abstract}
We study the problem of fair and efficient allocation of a set of indivisible goods to agents with additive valuations using the popular fairness notions of envy-freeness up to one good (EF1) and equitability up to one good (EQ1) in conjunction with Pareto-optimality (PO). There exists a pseudo-polynomial time algorithm to compute an EF1+PO allocation and a non-constructive proof of the existence of allocations that are both EF1 and fractionally Pareto-optimal (fPO), which is a stronger notion than PO. We present a pseudo-polynomial time algorithm to compute an EF1+fPO allocation, thereby improving the earlier results. Our techniques also enable us to show that an EQ1+fPO allocation always exists when the values are positive and that it can be computed in pseudo-polynomial time.

We also consider the class of $k$-ary instances where $k$ is a constant, i.e., each agent has at most $k$ different values for the goods. For such instances, we show that an EF1+fPO allocation can be computed in strongly polynomial time. When all values are positive, we show that an EQ1+fPO allocation for such instances can be computed in strongly polynomial time. Next, we consider instances where the number of agents is constant and show that an EF1+PO (likewise, an EQ1+PO) allocation can be computed in polynomial time. These results significantly extend the polynomial-time computability beyond the known cases of binary or identical valuations.

We also design a polynomial-time algorithm that computes a Nash welfare maximizing allocation when there are constantly many agents with constant many different values for the goods. Finally, on the complexity side, we show that the problem of computing an EF1+fPO allocation lies in the complexity class $\s{PLS}$. 
\end{abstract}

\section{Introduction}\label{sec:intro}
The problem of fair division was formally introduced by Steinhaus \cite{steinhaus} and has since been extensively studied in various fields, including economics and computer science \cite{Brams1996FairD,Moulin2003}. It concerns allocating resources to agents in a \textit{fair} and \textit{efficient} manner and has various practical applications such as rent division, division of inheritance, course allocation, and government auctions. Much of earlier work has focused on \textit{divisible} goods, which agents can share. In this setting, a prominent fairness notion is \textit{envy-freeness}~\cite{foleyEF,VARIAN1974-efpo}, which requires that every agent prefer their own bundle of goods to that of any other. 
On the other hand, when the goods are \textit{indivisible}, envy-free allocations need not even exist, e.g., in the simple case of one good and two agents. Other classical notions of fairness, like \textit{equitability} and \textit{proportionality}, may also be impossible to satisfy when goods are indivisible. However, fair division of indivisible goods remains an important problem since goods cannot always be shared and because it models several practical scenarios such as a course allocation \cite{othman2010coursealloc}. We refer the reader to the recent surveys \cite{walsh2020survey,amanatidis22survey} for other applications and recent results.

Since allocations satisfying standard fairness criteria fail to exist in the case of indivisible goods, several weaker fairness notions have been defined. A relaxation of envy-freeness called \textit{envy-freeness up to one good} (EF1) was defined by Budish \cite{budish2011approxCEEI}. An allocation is said to be EF1 if every agent prefers their own bundle over the bundle of any other agent after removing at most one good from the other agent's bundle. When the valuations of the agents for the goods are monotone, EF1 allocations always exist and can be computed in polynomial time \cite{lipton}.

The standard notion of economic efficiency is Pareto optimality (PO). An allocation is said to be PO if no other allocation makes an agent better off without making someone else worse off. A natural question is whether EF1 can be achieved with PO under additive valuations, which is the valuation class we focus on in this work. The concept of \textit{Nash welfare} provides a positive answer to this question. The Nash welfare is the geometric mean of the agents' utilities, and the allocation maximizing it achieves a tradeoff between efficiency and fairness. Caragiannis et al. \cite{caragiannis16nsw-ef1} showed that any maximum Nash welfare (MNW) allocation is EF1 and PO. For the special case of binary additive valuations, the MNW allocation can be computed in polynomial time \cite{darmann2014binary,barman2018binarynsw,nisarg20binaryonerule}. However, in general, the problem of computing the MNW allocation is APX-hard \cite{lee2015nsw-apx,garg2017nswhardness}. Moreover, it is not known if approximately Nash-optimal allocations retain the EF1 fairness guarantee, implying that approximation algorithms for MNW allocation, e.g., \cite{nguyen2014,cole2015nswapprox} may not be useful for computing an EF1+PO allocation.

Bypassing this barrier, Barman, Krishnamurthy, and Vaish \cite{Barman18FFEA} devised a pseudo-polynomial time algorithm that computes an allocation that is both EF1 and PO. They also showed that allocations that are both EF1 and \textit{fractionally} Pareto-optimal (fPO) always exist, where an allocation is said to be fPO if no fractional allocation exists that makes an agent better off without making anyone else worse off. They showed this result via a non-constructive convergence argument used in real analysis and did not provide an algorithm for computing such an allocation. Clearly, fPO is a stronger notion of economic efficiency, so the problem of computing EF1+fPO allocations is important.  Another reason to prefer fPO allocations over PO allocations in practice is that the former property admits efficient verification, whereas checking if an allocation is PO is known to be coNP-complete \cite{keijzer}. When a centralized entity is responsible for the allocation, all participants can efficiently verify if an allocation is fPO (and thus PO). However, in general, the same efficient verification is not possible for PO allocations. 

In this paper, we present a pseudo-polynomial time algorithm that computes an allocation that is EF1+fPO. Not only does this improve the result of Barman et al. \cite{Barman18FFEA}, but it also provides other interesting insights. We consider the class of $k$-ary instances. i.e., each agent has at most $k$ different (agent-specific) values for the goods. Our analysis shows that an EF1+fPO allocation can be found in polynomial time for $k$-ary instances when $k$ is a constant. Our result becomes especially interesting because computing the MNW allocation remains NP-hard for such instances \cite{lee2015nsw-apx}, even for $k=3$ \cite{amanatidis2020mnwefx}. Further, at present, this is the \textit{only} class apart from binary or identical valuations for which EF1+fPO allocations are polynomial time computable.

While $k$-ary instances are interesting theoretically to understand the limits of tractability in computing fair and efficient allocations, they are also relevant from a practical perspective. Eliciting agents' values for goods is often tricky, as agents may not be able to assert exactly what values they have for different goods. A simple protocol that the entity in charge of the allocation can do is to ask each agent to ``rate" the goods using a few (constantly many) values. Based on these responses, the valuations of the agents can be established.

Our results also extend to the fairness notion of \textit{Equitability up to one good} (EQ1), which is a generalization of the classical fairness notion of equitability. An allocation is said to be EQ1 if the utility an agent gets from her bundle is no less than the utility any other agent gets after removing one specific good from their bundle. Using similar techniques to that of Barman et al. \cite{Barman18FFEA}, a pseudo-polynomial time algorithm to compute an EQ1+PO allocation was developed by Freeman et al. \cite{freeman2019eqxpo} when all the values are positive. We show the stronger result that EQ1+fPO allocations always exist for positive-valued instances and can be computed in pseudo-polynomial time. Our techniques also show that for $k$-ary instances with positive values where $k$ is a constant, an allocation that is EQ1 and fPO can be computed in polynomial time. 

We next show that for constant $n$, an EF1+PO allocation can be computed in time polynomial in the number of goods. This result is significant since the number of agents $n$ is constant in many practical applications. In contrast, computing the MNW allocation remains NP-hard for $n=2$. Our techniques also show that for constant $n$, an EQ1+PO allocation can also be computed in polynomial time. 

Further, for $k$-ary instances with constant $n$ and $k$, we show that many fair division problems, including computing the MNW allocation, have polynomial time complexity. This improves the result of Bliem et al. \cite{bliem2016param}, showing that the EF+PO problem is tractable in this case.

We also make progress on the complexity front. We prove that the problem of computing an EF1+fPO lies in the complexity class Polynomial Local Search ($\s{PLS}$). For this, we carefully analyze our algorithm computing an EF1+fPO allocation and show that it has the structure of a local-search problem. Finally, we remark that our techniques also improve the results of Chakraborty et al. \cite{chak2020wef1} and Freeman et al. \cite{freeman2020chores} for the problems of computing weighted-EF1+fPO allocations of goods and EQ1+fPO allocations of \emph{chores}, respectively.

We summarize our results in the following table. A preliminary version of the present work appeared at  AAAI 2021 \cite{GargM21}.

\begin{table}[t]
\label{tab:results}
\centering
\begin{tabular}{|p{0.25\linewidth}||c|c|}\hline
	Instance type & EF1+fPO & EQ1+fPO$^*$  \\\hhline{|=||=|=|}
	constant $n$ & $\poly{m}$ (Theorem~\ref{thm:const-n-ef1po}) & $\poly{m}$ (Theorem~\ref{thm:const-n-eq1po}) \\\hline
	$k$-ary with constant $k$ & $\poly{m,n}$ (Theorem~\ref{thm:ef1po-sparse}) & $\poly{m,n}$ (Theorem~\ref{thm:eq1po-sparse}) \\\hline
	general additive & $\poly{n,m,v_{max}}$ (Theorem~\ref{thm:ef1po-pseudopoly}) & $\poly{n,m,v_{max}}$  (Theorem~\ref{thm:eq1po-pseudopoly}) \\\hline
\end{tabular}
\caption{Best-known algorithm run-times for the EF1+fPO and EQ1+fPO problems classified according to instance type. Here $n$ and $m$ denote the number of agents and goods, respectively, and $v_{max}$ denotes the maximum utility value. Agents have additive valuations. $^*$Results for the EQ1+fPO problem apply to positive instances.}
\end{table}

\subsection{Related Work}
Since the fair division literature is too vast, we refer the reader to surveys \cite{walsh2020survey,amanatidis22survey} for results on other fairness notions like proportionality. Below, we mention works related to fairness notions considered in this paper.

\paragraph{EF1+PO for goods.} Barman, Krishnamurthy, and Vaish \cite{Barman18FFEA} devised a pseudo-polynomial time algorithm that computes an allocation that is both EF1 and PO. This algorithm runs in time $\poly{n,m,v_{max}}$, where $n$ is the number of agents, $m$ is the number of items, and $v_{max}$ is the maximum utility value. Their algorithm first perturbs the values to a desirable form and then computes an EF1 and fPO allocation for the perturbed instance. Their approach is via \textit{integral market-equilibria}, which guarantees fPO at every step. The \textit{spending} of an agent, which is the sum of prices of the goods she owns in the equilibrium, works as a proxy for her utility. The returned allocation is approximately-EF1 and approximately-PO for the original instance, which, for a fine enough approximation, is EF1 and PO. Our algorithm proceeds similarly to their algorithm, with one main difference being that we do not need to consider any approximate instance and can work directly with the given valuations. Our algorithm returns an allocation that is not only PO but is fPO. Another key difference is the \textit{run-time analysis}: while their analysis relies on bounding the number of steps using arguments about \textit{prices}, our analysis is more direct and works with the \textit{values}. This allows us to prove a general result (Theorem~\ref{thm:ef1po-main}), a consequence of which is polynomial run-time for $k$-ary instances with constant $k$. Directly, such a conclusion cannot be drawn from the analysis of Barman et al. Another technical difference is that we raise the prices of multiple components of \emph{least spenders} simultaneously, unlike in \cite{Barman18FFEA}, and our analysis is arguably simpler than theirs. Barman et al. showed that there is a non-deterministic algorithm that computes an EF1+fPO allocation since checking if an allocation is fPO can be done efficiently. In contrast, we present a deterministic algorithm computing an EF1+fPO allocation, albeit with worst-case pseudopolynomial run time. Garg and Murhekar \cite{garg2021sagt} showed that allocations that are EFX (envy-free up to the removal of \textit{any} good) and fPO exist and can be computed in polynomial time for a subclass of 2-ary instances called bivalued instances. 

\paragraph{EF1+PO for chores.} Recently, the existence of EF1+PO allocations of \textit{chores} has been shown for special classes, though the question of existence in its full generality remains open. The existence and polynomial time computation of an EF1+fPO allocation for chores is known for (i) Bivalued instances, shown by \cite{GMQ21bivaluedchores} and \cite{ebadian2021bivaluedchores} independently, (ii) two types of chores, \cite{aziz2022twotypes}, (iii) $n = 3$ agents \cite{GMQ22ijcai}, and (iv) two types of agents \cite{GMQ22ijcai}.

\paragraph{EQ1+PO.} Freeman et al. \cite{freeman2019eqxpo} presented a pseudo-polynomial time algorithm for computing an EQ1+PO allocation for instances with positive values. Since they consider an approximate instance, too, their algorithm does not achieve the guarantee of fPO. They also show that the leximin solution, i.e., the allocation that maximizes the minimum utility and, subject to this, maximizes the second minimum utility, and so on, is EQX (a stronger requirement than EQ1) and PO for positive utilities. However, a simple reduction from the partition problem shows that computing a leximin solution is intractable. For positive bivalued instances, Garg and Murhekar \cite{garg2021sagt} showed that an EQX+fPO allocation can be computed in polynomial time. For chores, Freeman et al. \cite{freeman2020chores} showed that EQ1+PO allocations can be computed in pseudo-polynomial time.

\paragraph{Other results.} Bredereck et al. \cite{knop2019param} presented a framework for fixed-parameter algorithms for many fairness concepts, including EF1+PO, parameterized by $n+v_{max}$. Our results improve their findings and show that the problem has polynomial time complexity for the cases of (i) constant $n$, (ii) constant number of utility values, (iii) $v_{max}$ bounded by $\poly{m,n}$. 

\section{Preliminaries}
Let $\N^+$ denote the set of positive integers. For $t\in\N^+$, let $[t]$ denote $\{1,\dots,t\}$.

\paragraph{Problem setting.}
A \textit{fair division instance} is a tuple $(N,M,V)$, where $N = [n]$ is a set of $n\in\N^+$ agents, $M = [m]$ is the set of $m\in\N^+$ indivisible items, and $V = \{v_1,\dots,v_n\}$ is a set of utility functions, one for each agent $i\in N$. Each utility function $v_i : M \rightarrow \N^+$ is additive and is specified by $m$ numbers $v_{ij}\in\N$, one for each good $j\in M$, denoting the value agent $i$ has for good $j$. Additivity of the valuation functions implies that for every agent $i \in N$, and $S \subseteq M$, $v_i(S) = \sum_{j\in S} v_{ij}$. Further, we assume that for every good $j$, there is some agent $i$ such that $v_{ij} > 0$. Otherwise, the good $j$ for which $v_{ij}=0$ for all $i$ can be assigned to any agent arbitrarily. Note that we can in general work with rational values since they can be scaled to make them integral.

We call a fair division instance $(N, M, V)$ a:
\begin{enumerate}
\item \textit{Binary} instance if for all $i\in N$ and $j\in M$, $v_{ij}\in \{0,1\}$.
\item $k$-\textit{ary} instance, if $\forall i\in N$, $|\{v_{ij} : j\in M\}| \le k$.
\item \textit{Positive-valued} instance if $\forall i\in N$, $\forall j\in M$, $v_{ij} > 0$.
\end{enumerate}
\noindent Note that the class of $k$-ary instances generalizes the class of $k$-valued instances as defined by Amanatidis et al. \cite{amanatidis2020mnwefx}, in which all the values belong to a $k$-sized set, whereas we allow each agent to have $k$ different values for the goods.

\paragraph{Allocation.} An \textit{allocation} $\x$ of goods to agents is a $n$-partition of the goods $\x_1, \dots, \x_n$, where agent $i$ is allotted $\x_i \subseteq M$, and gets a total value of $v_i(\x_i)$. A \textit{fractional allocation} $\x \in [0,1]^{n\times m}$ is a fractional assignment such that for each good $j\in M$, $\sum_{i\in N} x_{ij} \le 1$. Here, $x_{ij}\in[0,1]$ denotes the fraction of good $j$ allotted to agent $i$.

For an agent $i\in N$, let $U_i$ be the number of different utility values $i$ can get in any allocation. Let $U = \max_{i\in N} U_i$.

\paragraph{Fairness notions.} An allocation $\x$ is said to be:
\begin{enumerate}
\item \textit{Envy-free up to one good} (EF1) if for all $i,h \in N$, there exists a good $j\in \x_h$ s.t. $v_i(\x_i) \ge v_i(\x_h \setminus \{j\})$.
\item \textit{Equitable up to one good} (EQ1) if for all $i,h \in N$, there exists a good $j\in \x_h$ s.t. $v_i(\x_i) \ge v_h(\x_h \setminus \{j\})$.
\end{enumerate}

\paragraph{Pareto-optimality.} An allocation $\y$ dominates an allocation $\x$ if $v_i(\y_i) \ge v_i(\x_i), \forall i$ and there exists $h$ s.t. $v_h(\y_h) > v_h(\x_h)$. An allocation is said to be \textit{Pareto optimal} (PO) if no allocation dominates it. Further, an allocation is said to be fractionally PO (fPO) if no fractional allocation dominates it. Thus, an fPO allocation is PO, but not vice-versa.

We remark that the EF1 and PO/fPO properties are scale-invariant. That is, if an allocation $\x$ is EF1+fPO for an instance $I$, then $\x$ is also EF1+fPO for an instance $I'$ where the valuation function of any agent $i$ is multiplied by a positive constant $\lambda_i$, i.e., $v'_{ij} = \lambda_i\cdot v_{ij}$ for every $i\in N$ and $j\in M$.

\paragraph{Maximum Nash Welfare.} The Nash welfare of an allocation $\x$ is given by $\NW(\x) = \big(\Pi_{i\in N} v_i(\x_i)\big)^{1/n}$. An allocation that maximizes the NW is called an MNW allocation.

\paragraph{Fisher markets.} A Fisher market or a \textit{market instance} is a tuple $(N,M,V,e)$, where the first three terms are interpreted as before, and $e=\{e_1,\dots,e_n\}$ is the set of agents' budgets, where $e_i \ge 0$, for each $i\in N$. In this model, goods can be allocated fractionally. Given prices $\p=(p_1, \dots, p_m)$ of goods, each agent aims to obtain the set of goods that maximizes her total value subject to her budget constraint.

A \textit{market outcome} is a (fractional) allocation $\x$ of the goods to the agents and a set of prices $\p$. The spending of an agent $i$ under the market outcome $(\x,\p)$ is given by $\p(\x_i) = \sum_{j\in M} p_jx_{ij}$. For an agent $i$, we define the \textit{bang-per-buck} ratio $\alpha_{ij}$ of good $j$ as $v_{ij}/p_j$, and the \textit{maximum bang-per-buck} (MBB) ratio $\alpha_i = \max_{j} \alpha_{ij}$. We define $\MBB_i = \{j\in M : v_{ij}/p_j = \alpha_i \}$, called the \textit{MBB-set}, to be the set of goods that give MBB to agent $i$ at prices $\p$. A market outcome $(\x,\p)$ is said to be \textit{`on \MBB'} if for all agents $i$ and goods $j$, $x_{ij} > 0$ implies $j\in \MBB_i$. For integral $\x$, this means $\x_i\subseteq \MBB_i$ for all $i\in N$. 

A market outcome $(\x,\p)$ is said to be a \textit{market equilibrium} if:
\begin{enumerate}
\item[(i)] the market clears, i.e., all goods are fully allocated. Thus, for all $j$, $\sum_{i\in N} x_{ij} = 1$,
\item[(ii)] budget of all agents is exhausted, for all $i\in N$, $\sum_{j\in M} x_{ij} p_j = e_i$, and
\item[(iii)] agents only spend money on goods that give them maximum bang-per-buck, i.e., $(\x,\p)$ is on MBB. 
\end{enumerate}

Given a market outcome $(\x,\p)$ with $\x$ integral, we say it is \textit{price envy-free up to one good} (pEF1) if for all $i,h\in N$ there is a good $j\in \x_h$ such that $\p(\x_i) \ge \p(\x_h\setminus \{j\})$. For integral market outcomes on MBB, the pEF1 condition implies the EF1 condition. 

\begin{lemma}\label{lem:pEF1impliesEF1}
Let $(\x,\p)$ be an integral market outcome on MBB. If $(\x,\p)$ is pEF1 then $\x$ is EF1 and fPO.
\end{lemma}
\begin{proof}
We first show that $(\x,\p)$ forms a market equilibrium for the Fisher market instance $(N,M,V,e)$, where for every $i\in N$, $e_i = \p(\x_i)$. It is easy to see that the market clears and the budget of every agent is exhausted. Further, $\x$ is on MBB as per our assumption. Now the fact that $\x$ is fPO follows from the First Welfare Theorem \cite{mas1995microeconomic}, which shows that for any market equilibrium $(\x,\p)$, the allocation $\x$ is fPO. 

Since $(\x,\p)$ is pEF1, for all pairs of agents $i,h \in N$, there is some good $j\in \x_h$ s.t. $\p(\x_i)\ge \p(\x_h \setminus \{j\})$. Since $(\x,\p)$ is on MBB, $\x_i \subseteq \MBB_i$. Let $\alpha_i$ be the MBB-ratio of $i$ at the prices $\p$. By definition of MBB, $v_i(\x_i) = \alpha_i \p(\x_i)$, and $v_i(\x_h\setminus \{j\}) \le \alpha_i \p(\x_h\setminus \{j\})$. Combining these, we get that $\x$ is EF1.
\end{proof}

Lemma~\ref{lem:pEF1impliesEF1} suggests that in order to obtain an EF1+PO allocation, it suffices to compute an integral market outcome $(\x,\p)$ that is pEF1. In an outcome $(\x,\p)$, we call an agent $i$ with minimum $\p(\x_i)$ a \textit{least spender} (LS) agent. If an outcome $(\x,\p)$ is not pEF1, then there must be an agent $h$ such that any LS $i$ pEF1-envies $h$. That is, $\p(\x_h\setminus \{j\}) > \p(\x_i)$, for every $j\in\x_h$ and any LS $i$. We call such an agent $h$ a \textit{pEF1-violator} in $(\x,\p)$.

Given a price vector $\p$, we define the MBB graph to be the bipartite graph $G = (N,M,E)$ where for an agent $i$ and good $j$, $(i,j)\in E$ iff $j\in \MBB_i$. Such edges are called \textit{MBB} edges. Given an accompanying allocation $\x$, we supplement $G$ to include \textit{allocation edges}, an edge between $i$ and $j$ if $j\in\x_i$.  

For agents $i_0,\dots,i_{\ell}$ and goods $j_1,\dots,j_{\ell}$, a path $P = (i_0, j_1, i_1, j_2, \dots, j_{\ell}, i_{\ell})$ in the MBB graph, where for all $1\le \ell' \le \ell$, $j_{\ell'} \in \x_{i_{\ell'}}\cap \MBB_{i_{\ell'-1}}$, is called a \textit{special} path. We define the \textit{level} $\lambda(h; i_0)$ of an agent $h$ w.r.t. $i_0$ to be half the length of the shortest special path from $i_0$ to $h$, and to be $n$ if no such path exists. A path $P = (i_0, j_1, i_1, j_2, \dots, j_{\ell}, i_{\ell})$ is an \textit{alternating path} if it is special, and if $\lambda(i_0; i_0) < \lambda(i_1; i_0) < \dots < \lambda(i_\ell; i_0)$, i.e., the path visits agents in increasing order of their level w.r.t. $i_0$. Further, the edges in an alternating path alternate between allocation and MBB edges. Typically, we consider alternating paths starting from a least spender (LS) agent. We use alternating paths to reduce the pEF1-envy of agent $i$ towards agent $h$ by transferring goods along the alternating path $(i_0, j_1, i_1, j_2, \dots, j_{\ell}, i_{\ell})$, i.e., by transferring $j_\ell$ to $i_{\ell-1}$, etc. The structure of an alternating path ensures that transfers are done along MBB edges, implying that resulting allocations remain on MBB and hence preserve the fPO property of the allocation. Figure~\ref{fig:altpath} provides an example of an alternating path from agent $i_0$ to agent $i_3$ involving agents $i_1$ and $i_2$ and goods $j_1,j_2$, and $j_3$.

\begin{figure}[h]
    \centering
    \includegraphics[scale=0.6]{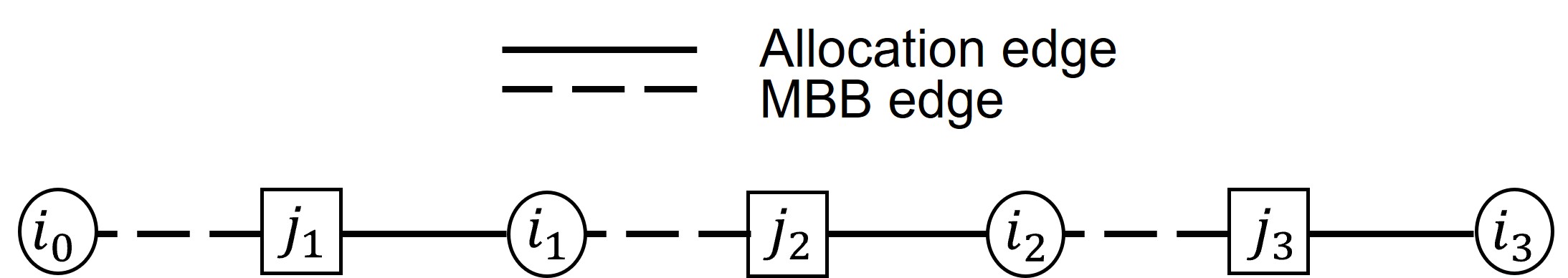}
    \caption{Example of an alternating path}
    \label{fig:altpath}
\end{figure}

\begin{definition}[Component $C_i$ of a least spender $i$]\label{def:component}
For a least spender $i$, define $C_i^{\ell}$ as the set of all goods and agents that lie on alternating paths of length $\ell$. Call $C_i = \bigcup_{\ell} C_i^{\ell}$ the \textit{component} of $i$, the set of all goods and agents reachable from the least spender $i$ through alternating paths.
\end{definition}

We now illustrate the terms introduced in the above definitions through an example.

\begin{example}\label{ex:example}\normalfont
Consider a fair division instance $(N,M,V)$ with three agents $\{a_1,a_2,a_3\}$ and five goods $\{g_1,\dots,g_5\}$. The values are given by:

\begin{table}[h]
    \centering
    \begin{tabular}{|c||c|c|c|c|c|}\hline
              & $g_1$ & $g_2$ & $g_3$ & $g_4$ & $g_5$ \\\hline
        $a_1$ & 6 & 4 & 0 & 0 & 0 \\
        $a_2$ & 0 & 4 & 2 & 5 & 0 \\
        $a_3$ & 4 & 3 & 1 & 4 & 2\\\hline
    \end{tabular}
    \caption{A fair division instance}
    \label{tab:example}
\end{table}

Consider the allocation $\x$ given by $\x_1 = \{g_1,g_2\}, \x_2 = \{g_3,g_4\}, \x_3 = \{g_5\}$ with associated price vector $\p = (6,4,2,5,2)$ defining the prices of the goods. The agents spending are given by $\p(\x_1) = 10$, $\p(\x_2) = 7$, and $\p(\x_3) = 2$. It can be checked that the MBB ratios of all agents are one and that the allocation is on MBB. Figure~\ref{fig:ex-mbb} describes the MBB graph associated with the allocation $(\x,\p)$. In the graph, $(a_2, g_2, a_1)$ is an alternating path. The least spender (LS) is $a_3$. The component of the LS is the agent $a_3$ and the good $g_5$.

\begin{figure}[h]
    \centering
    \includegraphics[scale=0.6]{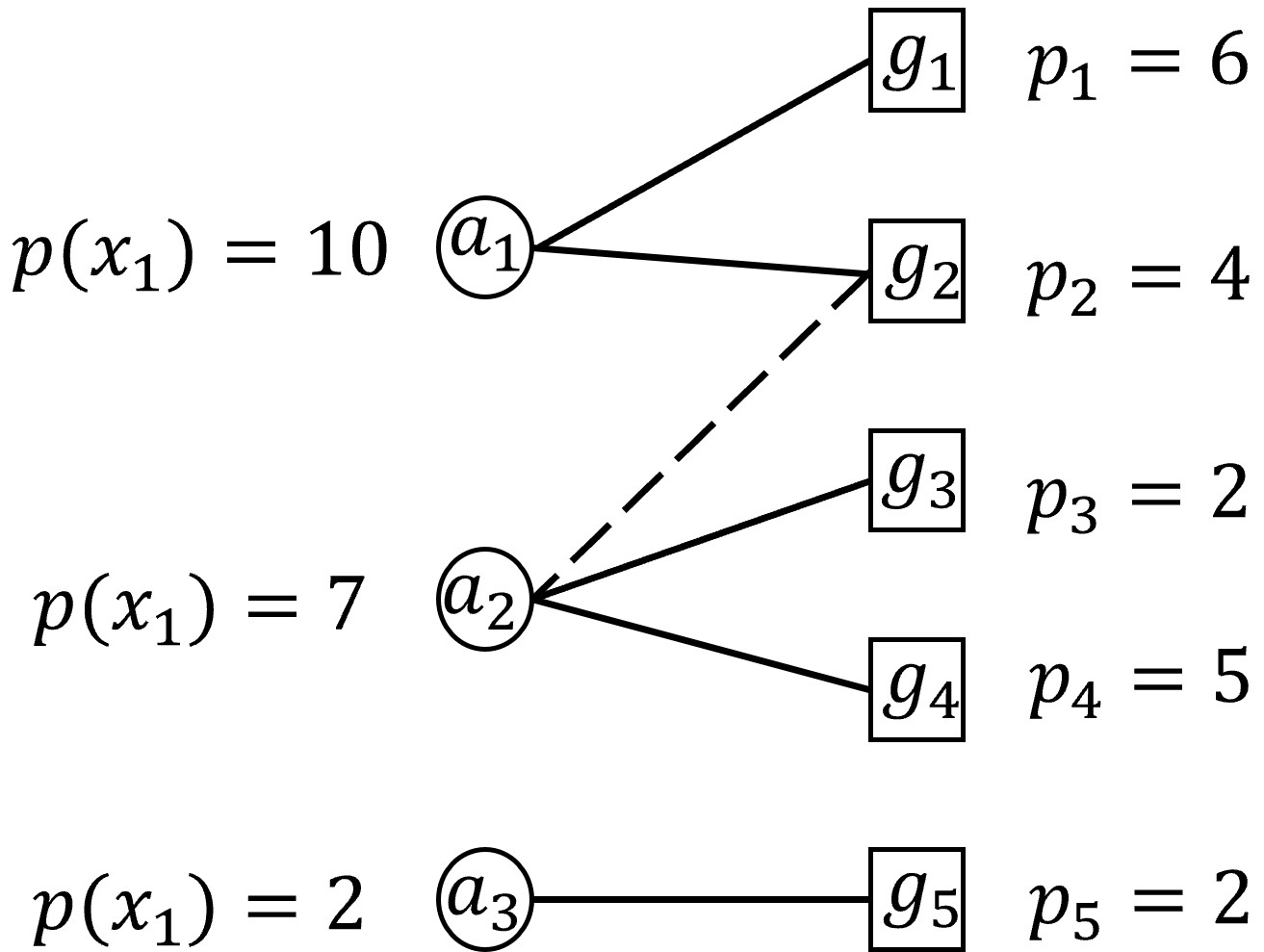}
    \caption{An allocation $(\x,\p)$ and the associated MBB graph}
    \label{fig:ex-mbb}
\end{figure}
\end{example}

\section{Finding EF1+fPO allocations of goods}\label{sec:ef1po}
We now present the main algorithm of our paper. We show that given a fair division instance $(N,M,V)$, Algorithm~\ref{alg:ef1po} returns an allocation $\x$ that is EF1 and fPO and terminates in time $\poly{n,m,U}$. Recall that $U$ denotes the maximum number of distinct utility values an agent can achieve in any allocation.

Algorithm~\ref{alg:ef1po} starts with a welfare maximizing integral allocation $(\x,\p)$, where $p_j = v_{ij}$  for $j \in \x_i$. If the allocation $\x$ is not EF1, then the outcome $(\x,\p)$ cannot be pEF1. Thus, there must exist a pEF1-violator agent. Our algorithm aims to reduce pEF1-envy by transferring goods away from such pEF1-violator agents. Let $L$ denote the set of least spenders. We first explore if a pEF1-violator is reachable via alternating paths starting from some least spender (LS), i.e., belongs to the component $C_i$ of some LS $i$. If not, then we \textit{raise the prices} of all goods in $C_L$, the union of components of all least spenders, until either (i) a new MBB edge gets added from an agent $h\in C_L$ to a good $j\notin C_L$ (corresponding to a price-rise of $\gamma_1$), or (ii) the spending of an agent $h\notin C_L$ becomes equal to the spending of the agents in $L$, i.e., a new agent becomes an LS (corresponding to a price-rise of $\gamma_2$).

Otherwise, we choose the minimum-index least spender $i$ whose component $C_i$ contains a pEF1-violator. We find an alternating path $P = (i, j_1, i_1, \dots, j_{\ell}, i_{\ell} = h)$, from $i$ to a pEF1-violator $h$, and it must be the case that $\p(\x_h\setminus\{j_{\ell}\}) > \p(\x_i)$. We then \textit{transfer} the good $j_\ell$ from $h$ to $i_{\ell-1}$. The choice of the minimum-index LS is to break ties among all LSs whose component contains a pEF1-violator. This ensures that there are no back-to-back transfers, where good $j$ transferred from agent $h$ to $h'$ via an alternating path starting from LS $i$ and immediately in the next step from $h'$ to $h$ via an alternating path starting from another LS $i'$, as our algorithm will only choose one LS $\min\{i,i'\}$. Moreover, this ensures that while the component $C_i$ contains pEF1-violators for the minimum-index LS $i$, we only perform transfers from the pEF1-violators in $C_i$.

Our algorithm repeatedly performs either price-rise steps or transfer steps until the outcome $(\x,\p)$ becomes pEF1. Before showing that the algorithm always terminates and analyzing its run-time, we illustrate its execution through an example.

\begin{algorithm}[t]
\caption{Computing an EF1+fPO allocation of goods}\label{alg:ef1po}
\textbf{Input:} Fair division instance
$(N,M,V)$\\
\textbf{Output:} An integral EF1+PO allocation $\x$
\begin{algorithmic}[1]
\State $(\x,\p) \gets$ Initial welfare maximizing integral market allocation, where $p_j = v_{ij}$ for $j \in \x_i$.
\While{$(\x, \p)$ is not pEF1}
\State $L \gets \{k \in N : k \in \arg\min_{h\in N} \p(\x_h)\}$ \Comment{set of LS}
\If{for all $i\in L$, there is no pEF1-violator in $C_i$} \Comment{Perform price-rise step}
\State $\gamma_1 = \min_{h\in C_L \cap N, j\in M\setminus C_L} \frac{\alpha_h}{v_{hj}/p_j}$ \Comment{Factor by which prices of goods in $C_L$ are raised until a new MBB edge appears from an agent in $C_L$ to a good outside $C_L$}
\State $\gamma_2 = \min_{i\in L, h\in N\setminus C_L} \frac{\p(\x_h)}{\p(\x_i)}$ \Comment{Factor by which prices of goods in $C_L$ are raised until a new agent outside $C_L$ becomes a new LS}
\State $\beta = \min(\gamma_1,\gamma_2)$ \Comment{Prise rise factor}
\For{$j \in C_L\cap M$} \Comment{Raise prices of goods in $C_L$}
\State $p_j \gets \beta p_j$
\EndFor
\Else \Comment{Perform transfer step}
\State $i \gets \min\{k\in L: \text{ there is a pEF1-violator in $C_k$}\}$ \Comment{LS with smallest index}
\State Let $(i, j_1, i_1, \dots, i_{\ell-1}, j_\ell, i_\ell)$ be an alternating path from $i$ to pEF1-violator $i_\ell$
\State Transfer $j_\ell$ from $i_\ell$ to $i_{\ell-1}$
\EndIf
\EndWhile
\Return $\x$
\end{algorithmic}
\end{algorithm}

\begin{example} \normalfont
We revisit the instance in Example~\ref{ex:example}, and begin with welfare maximizing allocation described in Example~\ref{ex:example}, i.e., $(\x,\p)$ where $\x_1 = \{g_1,g_2\}, \x_2 = \{g_3,g_4\}, \x_3 = \{g_5\}$ and $\p = (6,4,2,5,2)$. Allocation $\x$ is not EF1 since $v_3(\x_3) = 2 < 3 = v_3(\x_1 \setminus g_1)$. Indeed, $(\x,\p)$ is not pEF1 either, with $a_3$ pEF1-envying $a_1$, i.e., $a_1$ is a pEF1-violator. Since there is no alternating path from $a_3$ (the least spender) to $a_1$ (pEF1-violator), Algorithm~\ref{alg:ef1po} increases the price of $g_5$. A new MBB edge appears from $a_3$ to $g_4$ on a price-rise factor of $\gamma_1 = \frac{v_{35}/p_{5}}{v_{34}/p_4} = 1.25$, and a new LS appears on a price-rise factor of $\gamma_2 = \p(\x_2)/\p(\x_3) = 3.5$. Thus, the price of $g_5$ will be raised by $\gamma = \min(\gamma_1,\gamma_2) = 1.25$, i.e., the new price is $p_5 = 2.5$, and there an MBB edge from $a_3$ to $g_4$ is now present.

Since the allocation is not yet pEF1, Algorithm~\ref{alg:ef1po} checks if a pEF1-violator ($a_3$) belongs to the component of an LS ($a_1$). Since this is the case, Algorithm~\ref{alg:ef1po} identifies the shortest alternating path $(a_3,g_4,a_2,g_2,a_1)$ and makes the transfer of $g_2$ from $a_1$ to $a_2$. The new allocation is now $\x'$ given by $\x'_1 = \{g_1\}, \x'_2 = \{g_2,g_3,g_4\}, \x'_3 = \{g_5\}$. Since the allocation is still not pEF1 as $a_3$ pEF1-envies $a_2$, the algorithm transfers $g_4$ from the pEF1-violator $a_2$ to $a_3$. This results in the allocation $\x''$ given by $\x''_1 = \{g_1\}, \x''_2 = \{g_2,g_3\}, \x''_3 = \{g_4,g_5\}$, which is pEF1, and thus Algorithm~\ref{alg:ef1po} terminates with an EF1+PO allocation.
\end{example} 

We now proceed to analyze the run-time of the Algorithm~\ref{alg:ef1po}. We will use the terms \textit{time-step} or \textit{iteration} interchangeably to denote either a transfer or a price-rise step. We say `at time-step $t$' to refer to the state of the algorithm \textit{just before} the event at $t$ happens. We denote by $(\x^{t},\p^t)$ the allocation and price vector at time-step $t$. First, we note that: 

\begin{lemma}\label{lem:alloc-on-mbb}
At any time-step $t$, $(\x^t,\p^t)$ is on MBB.
\end{lemma}
\begin{proof}
We want to show that at every iteration of the algorithm, every agent owns goods from their MBB set. To see this, notice that this is the case in the algorithm's initial allocation in Line 1. Suppose we assume inductively that at iteration $t$, in the corresponding allocation $\x$, every agent buys MBB goods. We ensure that the goods are transferred along MBB edges, and thus, no transfer step causes the MBB condition of any agent to be violated. Consider a price-rise step, in which the prices of all goods in $C_L$, the component of the least spenders $L$, are increased by a factor of $\beta$. Since prices of these goods are raised, no agent $h \notin C_L$ prefers these goods over their own goods, and hence, they continue to own goods from their MBB set. For any agent $h\in C_L$ and a good $j\notin C_L$, the bang-per-buck that $j$ gets from $h$ before the price-rise is strictly less than her MBB ratio since $j$ is not in the MBB set of $h$. Further, we never raise these prices beyond the creation of a new MBB edge from any agent $h \in C_L$ to some agent $j\in C_L$. Thus, the MBB condition is not violated for any agent $h\in C_L$. 

In conclusion, the MBB condition is never violated for any agent throughout the algorithm, so the allocation is always on MBB.
\end{proof}

If Algorithm~\ref{alg:ef1po} terminates, then the final outcome $(\x,\p)$ is pEF1. Since it is also on MBB, by Lemma~\ref{lem:pEF1impliesEF1}, $\x$ is EF1+fPO. We now proceed towards the run-time analysis of Algorithm~\ref{alg:ef1po}. First, we observe that prices are raised only until the spending of a new agent becomes equal to the spending of the least spenders.
\begin{lemma}\label{lem:LSspending}
The spending of the least spender(s) does not decrease as the algorithm progresses. Further, at any price-rise event $t$ with price-rise factor $\beta$, the spending of the least spender(s) increases by a factor of $\beta$.
\end{lemma}
\begin{proof}
The spending of a least spender can decrease if the least spender loses a good. This cannot happen since the only agents who lose goods are pEF1-violators, whose spending after removing one good exceeds that of the least spender. Moreover, if a good $j$ is transferred from a pEF1-violator $h$ directly to a LS $i$, then we have $\p(\x_h\setminus \{j\}) > \p(\x_i)$. Thus, even if $h$ becomes a new LS after the transfer, the spending of the LS has only increased. Finally, note that in a price-rise step with price-rise factor $\beta$, the prices of all goods in $C_L$ are increased by a factor of $\beta$. Thus, the spending of every agent in $C_L$, including the least spenders, increases by a factor of $\beta$.
\end{proof}

\noindent Next, we argue:
\begin{lemma}\label{lem:swap-poly}
The number of iterations with the same set of least spenders is $\poly{m,n}$. 
\end{lemma}
\begin{proof}
Let us fix a set $L$ of least spenders. We first argue that there can be at most $n$ price-rise steps without any change in $L$. This is because of the following. Since $L$ is unchanged, every price-rise step can only create a new MBB edge since a new LS will change $L$. Thus, each price-rise step causes a new agent to be added to $C_L$. Since there are at most $n$ agents outside $C_L$ at any given iteration, there can only be $n$ price-rise steps.

We now argue that the number of transfer steps with the same set of LS between any two price steps is at most $\poly{m,n}$. Recall that our algorithm always uses the minimum-index least spender whose component contains a pEF1-violator. Because of this, all transfers are performed in a component $C_i$ before moving to another component $C_{\ell}$, for agents $i,\ell \in L$ with $i<\ell$.  Thus, if a transfer from a pEF1-violator agent $h$ to LS $\ell$ is performed, then $h$ can not be in $C_i$ at the time of transfer. \cite{Barman18FFEA} used a potential function argument to show that the number of consecutive transfer steps with a fixed LS agent is at most $m\cdot n^2$. Thus, with our observation above, we can conclude that the number of consecutive transfers with the same set of LS is at most $m\cdot n^3 = \poly{m,n}$.
\end{proof} 

The next lemma is key. We argue that between the time steps at which an agent $i$ ceases to be an LS and subsequently becomes an LS again, her utility strictly increases.

\begin{lemma}\label{lem:utilinc}
Let $t_0$ be a time-step where agent $i$ ceases to be an LS, and let $t_\ell$ be the first subsequent time step just after which $i$ becomes the LS again. Then:
\[ v_i(\x_i^{t_\ell+1}) > v_i(\x_i^{t_0}) \]
Note here that $v_i(\x_i^{t_0})$ is the utility of agent $i$ \emph{just before} time-step $t_0$, and $v_i(\x_i^{t_\ell+1})$ her utility \emph{just after} time-step $t_\ell$.
\end{lemma}
\begin{proof}
From Lemma~\ref{lem:LSspending}, since $i$ ceases to be an LS after time-step $t_0$, $i$ must have received some good $j$ at time step $t_0$. Since $j\in \MBB_i$ at $t_0$, $v_{ij} > 0$. Suppose $i$ does not lose any good in any subsequent iterations until $t_\ell$, then $\x^{t_\ell+1}_i \supseteq \x^{t_0}_i \cup \{j\}$, and hence $v_i(\x_i^{t_\ell+1}) \ge v_i(\x_i^{t_0}\cup \{j\}) = v_i(\x_i^{t_0}) + v_{ij} > v_i(\x_i^{t_0})$, using additivity of valuations.

On the other hand, suppose $i$ does lose some goods between $t_0$ and $t_\ell$. Let $t_k \in (t_0,t_\ell]$ be the last time-step when $i$ loses a good, say $j'$. Let $t_1,\dots,t_{k-1}$ be time-steps (in order) between $t_0$ and $t_k$ when $i$ experiences price-rise, and $t_{k+1},\dots,t_{\ell-1}$ be time-steps (in order) between $t_k$ and $t_\ell$ when $i$ experiences price-rise, until finally after the iteration $t_\ell$ agent $i$ becomes the LS again. Let us define $\beta_t$ to be the price-rise factor at the time-step $t$. If $t$ is a price-rise step, $\beta_t > 1$, else we set $\beta_t = 1$. Hence $\beta_{t_1},\dots,\beta_{t_{k-1}},\beta_{t_{k+1}},\dots,\beta_{t_{\ell-1}}$ are price-rise factors at the corresponding events $t_1,\dots,t_{k-1},t_{k+1},\dots,t_{\ell-1}$ and are all greater than 1. If $t_\ell$ is a price-rise event, let the price-rise factor be $\beta_{t_\ell} > 1$; and if not let $\beta_{t_\ell} = 1$. Note that $t_k$ is not a price-rise event; hence, $\beta_{t_k} = 1$.

Using Lemma~\ref{lem:LSspending}, together with the fact that $i$ does not lose any good after $t_k$, we have:
\begin{equation}\label{eqn:sparse-term-1}
\p^{t_\ell+1}(\x_i^{t_\ell+1}) \ge  (\beta_{t_\ell}\beta_{t_{\ell-1}}\cdots\beta_{t_{k+1}})\p^{t_k}(\x^{t_k}_i\setminus\{j'\}).
\end{equation}
The above may not be equality because in addition to experiencing price-rises during $t_{k+1},\dots,t_\ell$, agent $i$ may also \textit{gain} some new good.
If $i_k$ is a LS at $t_k$, then for agent $i$ to lose the good $j'$ it must be the case that:
\begin{equation}\label{eqn:sparse-term-2}
\p^{t_k}(\x^{t_k}_i \setminus \{j'\}) > \p^{t_k}(\x^{t_k}_{i_k}).
\end{equation}
Let $i_t$ be a LS at time-step $t$. Then by repeatedly applying Lemma~\ref{lem:LSspending}, we get:
\begin{equation}\label{eqn:sparse-term-3}
\begin{aligned}
\p^{t_k}(\x^{t_k}_{i_k}) &\ge \beta_{t_k-1}p^{t_k-1}(\x^{t_k-1}_{i_{t_k-1}}) \\
&\ge\dots
\ge (\beta_{t_k-1}\beta_{t_k-2}\cdots\beta_1)\p^{t_0}(\x^{t_0}_{i_{t_0}}) \\
&\ge (\beta_{t_{k-1}}\beta_{t_{k-2}}\cdots\beta_{t_1})\p^{t_0}(\x^{t_0}_{i})\enspace, 
\end{aligned}
\end{equation}
where the last transition follows from the facts that (i) each $\beta_t \ge 1$, (ii) $\{t_1,\dots,t_{k-1}\}\subseteq \{1,\dots,t_k-1\}$, and (iii) $i_{t_0} = i$ since $i$ is a least spender at $t_0$. Putting \eqref{eqn:sparse-term-1}, \eqref{eqn:sparse-term-2} and \eqref{eqn:sparse-term-3} together, we get: \begin{equation}\label{eqn:sparse-term-4}
\p^{t_\ell+1}(\x_i^{t_\ell+1}) > (\Pi_{r=1}^\ell \beta_{t_r}) \p^{t_0}(\x^{t_0}_{i})\enspace . 
\end{equation}
Let $\alpha^t_i$ denote the MBB-ratio of $i$ at the time step $t$. Observe that in every price-rise event with price-rise factor $\beta$, the MBB ratio of any agent experiencing the price-rise decreases by $\beta$. Further, the MBB ratio of any agent does not change unless she experiences a price-rise step. Thus:
\begin{equation}\label{eqn:sparse-term-5}
    \alpha_i^{t_\ell+1} = \frac{\alpha_i^{t_0}}{(\beta_{t_{\ell}}\beta_{t_{\ell-1}}\cdots\beta_{t_{k+1}})(\beta_{t_{k-1}}\beta_{t_{k-2}}\cdots\beta_{t_1})}\enspace . 
\end{equation}
Therefore, using the fact that the allocation is on MBB edges, and with~\eqref{eqn:sparse-term-4} and~\eqref{eqn:sparse-term-5}, we have:
\begin{alignat*}{2}
v_i(\x_i^{t_\ell+1}) &= \alpha_i^{t_\ell+1} \p^{t_\ell+1}(\x^{t_\ell+1}_i) &\text{($\x^{t_\ell+1}$ is on MBB)}&\\
&> \frac{\alpha_i^{t_0}}{(\Pi_{r=1}^\ell \beta_{t_r})}(\Pi_{r=1}^\ell \beta_{t_r})\p^{t_0}(\x^{t_0}_{i}) &\text{(From~\eqref{eqn:sparse-term-4} and \eqref{eqn:sparse-term-5})}&\\
&= \alpha^{t_0}_i\p^{t_0}(\x_i^{t_0})
= v_i(\x_i^{t_0}), &\text{($\x^{t_0}$ is on MBB)}&
\end{alignat*}
as claimed.
\end{proof}

Using the above lemmas, we show:
\begin{lemma}\label{lem:ef1po-runtime}
Algorithm~\ref{alg:ef1po} terminates in time $\poly{n,m,U}$.
\end{lemma}
\begin{proof}
Consider any agent $i$. From Lemma~\ref{lem:utilinc}, it is clear that every time $i$ becomes the LS again her utility has strictly increased compared to her utility the last time she was an LS. The number of utility values that $i$ can have is $U_i$; hence, we conclude that the number of times she stops being an LS and becomes LS again is at most $U_i$. Since there are $n$ agents, and each agent $i$ can become the LS again at most $U_i$ times, we have that after $\poly{n, \max_{i\in N} U_i}$ changes in the set of least spenders, there will be no changes further in the set of least spenders. After this, in at most $n$ more price-rise steps, either the allocation becomes pEF1, or all agents get added to $C_L$ since no new agent becomes an LS on raising prices. Further, the number of transfers with the same set of least spenders is at most $\poly{m,n}$ (Lemma~\ref{lem:swap-poly}). This shows that Algorithm~\ref{alg:ef1po} terminates in time $\poly{n,m,U}$.
\end{proof}
Putting it all together, we conclude:
\begin{theorem}\label{thm:ef1po-main}
Let $I = (N,M,V)$ be a fair division instance. Then, an allocation that is both EF1 and fPO can be computed in time $\poly{n,m,U}$.
\end{theorem}

Observe that in any allocation and for any agent, the minimum utility is 0, and the maximum utility is $mv_{max}$, where $v_{max} = \max_{i,j} v_{ij}$. Since the utility values are integral, we have $U \le mv_{max} + 1$. Thus, Algorithm 1 computes an EF1+fPO allocation in pseudo-polynomial time.

\begin{theorem}\label{thm:ef1po-pseudopoly}
Given a fair division instance $I = (N,M,V)$, an allocation that is both EF1 and fPO can be computed in time $\poly{n,m,v_{max}}$, where $v_{max} = \max_{i,j} v_{ij}$. In particular, when $v_{max} \le \poly{m,n}$, an EF1+fPO allocation can be computed in $\poly{m,n}$ time.
\end{theorem}

The guarantee of EF1+fPO offered by our algorithm is stronger than the guarantee of EF1+PO provided by the algorithm of Barman et al. \cite{Barman18FFEA}. 
We next turn our attention to $k$-ary instances where $k$ is a constant. First, we observe that for such instances, the maximum number of different utility values any agent can get is at most $\poly{m}$.
\begin{lemma}\label{lem:sparse-poly-utils}
In a $k$-ary fair division instance $(N,M,V)$ with constant $k$, $U \le \poly{m}$.
\end{lemma}
\begin{proof}
For any agent $i$, let $\{v_i^\ell\}_{\ell \in [k]}$ be the different utility values $i$ has for the goods. In an allocation $\x$, let $m_i^\ell \in \Z_{\ge 0}$ be the number of goods in $\x_i$ with value $v_i^\ell$. Then, agent $i$'s utility is simply:
$v_i(\x_i) = m_i^1 v_i^1 + \dots + m_i^k v_i^k$. Since each $0 \le m_i^\ell \le m$, the number of possible utility values that $i$ can get in any allocation is at most $(m+1)^k$, which is $\poly{m}$ since $k$ is constant. Thus $U \le \poly{m}$.
\end{proof}

Therefore, using Lemma~\ref{lem:sparse-poly-utils}, Theorem~\ref{thm:ef1po-main} gives:

\begin{theorem}\label{thm:ef1po-sparse}
Given a $k$-ary fair division instance $I = (N,M,V)$ where $k$ is a constant, an allocation that is both EF1 and fPO can be computed in time $\poly{m,n}$.
\end{theorem}

\begin{remark}\noindent
\normalfont We note that our techniques can also be used to show that an allocation that is weighted-EF1 (wEF1) and fPO exists and can be computed in pseudo-polynomial time. Here, an allocation $\x$ is wEF1 if for all agents $i,h$ we have $\frac{v_i(\x_i)}{w_i} \ge \frac{v_i(\x_h\setminus j)}{w_h}$ for some $j\in \x_h$, where $w_k$ denotes the weight of agent $k$. We can modify Algorithm~\ref{alg:ef1po} for the weighted setting by considering weighted-LS instead of LS, as the agents with minimum $\frac{\p(\x_i)}{w_i}$, and performing transfer from agent $h$ if it violates weighted-pEF1, i.e., if $\frac{\p(\x_h\setminus j)}{w_h} > \frac{\p(\x_i)}{w_i}$ for every $j\in \x_h$ and a weighted-LS $i$. Lemmas~\ref{lem:alloc-on-mbb} and \ref{lem:swap-poly} still hold as before. Lemma~\ref{lem:LSspending} holds with weighted-spending instead of spending. With these modifications, the key  Lemma~\ref{lem:utilinc} can be shown as before, thus implying pseudo-polynomial run time for computing an EF1+fPO allocation in the weighted case as well.
\end{remark}

\section{Finding EQ1+fPO allocations of goods}\label{sec:eq1po}
We now show that Algorithm~\ref{alg:eq1po} finds an EQ1+fPO allocation given a fair division instance with positive values. We require the values to be positive because instances with zero values might not even admit an allocation that is EQ1+PO \cite{freeman2019eqxpo}. Algorithm~\ref{alg:eq1po} is similar to Algorithm~\ref{alg:ef1po}, except that it works with values instead of spendings of agents since we desire EQ1 allocations instead of EF1.

Algorithm~\ref{alg:eq1po} starts with a welfare maximizing integral allocation $(\x,\p)$, where $p_j = v_{ij}$  for $j \in \x_i$. We refer to an agent with the least utility as an LU agent and let $L$ be the set of LU agents. If the allocation $\x$ is not EQ1, then exist an EQ1-violator agent $h$, i.e., $v_h(\x_h\setminus \{j\}) > v_i(\x_i)$, for any $j\in \x_h$ and any LU agent $i$. Similar to Algorithm~\ref{alg:ef1po}, we first explore if such an EQ1-violator belongs to the component $C_i$ of some LU agent $i$. If not, \textit{we raise the prices} of all goods in $C_L$ until a new MBB edge gets added from an agent $h\in C_L$ to a good $j\notin C_L$.

Otherwise, we choose a minimum-index LU agent $i$ whose component $C_i$ contains an EQ1-violator. We find an alternating path $P = (i, j_1, i_1, \dots, j_{\ell}, i_{\ell} = h)$, from $i$ to a EQ1-violator $h$, and it must be the case that $v(\x_h\setminus\{j_{\ell}\}) > v(\x_i)$. We then \textit{transfer} the good $j_\ell$ from $h$ to $i_{\ell-1}$.

Thus, the algorithm performs price-rise or transfer steps while the allocation is not EQ1. If the algorithm terminates, the allocation must be EQ1 (Line 2). By arguments similar to Lemma~\ref{lem:alloc-on-mbb}, we can show that the outcome throughout the execution of the algorithm is always on MBB and hence is fPO.

\begin{algorithm}[t] 
\caption{Computing an EQ1+fPO allocation of goods}
\label{alg:eq1po}
\textbf{Input:} Positive-valued fair division instance
$(N,M,V)$\\
\textbf{Output:} An integral EQ1+PO allocation ${\x}$
\begin{algorithmic}[1]
\State $(\x,\p) \gets$ Initial welfare maximizing integral market allocation, where $p_j = v_{ij}$ for $j \in \x_i$.
\While{$\x$ is not EQ1}
\State $L\gets \{i \in N : i \in \s{argmin}_{h\in N} v_h(\x_h)\}$ \Comment{set of LU agents}
\If{for all $i\in L$, there is no EQ1-violator in $C_i$} \Comment{Perform price-rise step}
\State $\beta = \min_{h\in C_L\cap N, j\in M\setminus C_L} \frac{\alpha_h}{v_{hj}/p_j}$ \Comment{Factor by which prices of goods in $C_L$ are raised until a new MBB edge appears from an agent in $C_L$ to a good outside $C_L$}
\For{$j \in C_L\cap M$}
\State $p_j \gets \beta p_j$
\EndFor
\Else \Comment{Perform transfer step}
\State $i \gets \min\{k\in L: \text{there is an EQ1-violator in $C_k$}\}$ \Comment{LU with smallest index}
\State Let $(i, j_1, i_1, \dots, i_{\ell-1}, j_\ell, i_\ell)$ be an alternating path from $i$ to EQ1-violator $i_\ell$
\State Transfer $j_\ell$ from $i_\ell$ to $i_{\ell-1}$
\EndIf
\EndWhile
\Return $\x$
\end{algorithmic}
\end{algorithm}

We now prove that the algorithm terminates. Similar to Lemma~\ref{lem:LSspending}, we can argue that the utility of the LU agents(s) does not decrease in the algorithm's run. Further, similar to Lemma~\ref{lem:utilinc}, we can show the following key lemma:
\begin{lemma}\label{lem:utilinc-eq1}
Let $t_0$ be a time-step where agent $i$ ceases to be an LU agent, and let $t_\ell$ be the first subsequent time step just after which $i$ becomes the LU agent again. Then:
\[
v_i(\x_i^{t_\ell+1}) > v_i(\x_i^{t_0}). 
\]
\end{lemma}
\begin{proof}
From Lemma~\ref{lem:LSspending}, $i$ must have received some good $j$ at time step $t_0$. Since $j\in \MBB_i$ at $t_0$, $v_{ij} > 0$. Suppose $i$ does not lose any good in any subsequent iterations, then $\x^{t_\ell+1}_i \supseteq \x^{t_0}_i \cup \{j\}$, and hence $v_i(\x_i^{t_\ell+1}) > v_i(\x_i^{t_0})$. On the other hand, suppose $i$ does lose some goods. Let $t_k\in(t_0,t_\ell]$ be a subsequent time step where $i$ loses a good $j'$ for the last time. Since $i$ does not lose any good after $t_k$, we have:
\begin{equation}\label{eqn:sparse-term-1-eq1}
v_i(\x_i^{t_\ell+1}) \ge  v_i(\x^{t_k}_i\setminus\{j'\}).
\end{equation}
If $i_k$ is a LU at $t_k$, then for $i$ to lose $j'$ it must be that:
\begin{equation}\label{eqn:sparse-term-2-eq1}
v_i(\x^{t_k}_i\setminus \{j'\}) > v_{i_{t_k}}(\x^{t_k}_{i_k}).
\end{equation}
Finally, since the utility of an LU agent does not decrease: 
\begin{equation}\label{eqn:sparse-term-3-eq1}
v_{i_{t_k}}(\x^{t_k}_{i_k}) \ge v_i({\x^{t_0}_i}).
\end{equation}
Putting~\eqref{eqn:sparse-term-1-eq1}, \eqref{eqn:sparse-term-2-eq1} and \eqref{eqn:sparse-term-3-eq1} together, we get:
\[ v_i(\x_i^{t_\ell+1}) > v_i(\x_i^{t_0}) \]
as required.
\end{proof}

Using the above, as argued in Lemma~\ref{lem:ef1po-runtime}, we can show:
\begin{lemma}\label{lem:eq1po-runtime}
Algorithm~\ref{alg:eq1po} terminates in time $\poly{n,m,U}$.
\end{lemma}
\begin{proof}
Consider any agent $i$. From Lemma~\ref{lem:utilinc-eq1}, it is clear that every time an agent $i$ becomes a LU agent again, her utility has strictly increased compared to her utility the last time she was a LU agent. The number of utility values that $i$ can have is $U_i$, and hence we conclude that the number of times she stops being an LU agent and becomes an LU agent again is at most $U_i$. Since there are $n$ agents, and each agent $i$ can become the LU agent again at most $U_i$ times, we have that after $\poly{n, \max_{i\in N} U_i}$ changes in the set of LU agents, there will be no changes further in the set of LU agents. Further, using an analysis similar to Lemma~\ref{lem:swap-poly}, we can show that the number of transfers with the same set of LU agents is at most $\poly{m,n}$. This shows that Algorithm~\ref{alg:ef1po} terminates in time $\poly{n,m,U}$.
\end{proof}

We conclude:
\begin{theorem}\label{thm:eq1po-main}
Let $I = (N,M,V)$ be a positive-valued fair division instance. Then, an allocation that is both EQ1 and fPO can be computed in time $\poly{n,m,U}$.
\end{theorem}

As argued before, we have $U \le mv_{max} + 1$. This gives:

\begin{theorem}\label{thm:eq1po-pseudopoly}
Given a fair division instance $I = (N,M,V)$, an allocation that is EQ1 and fPO can be computed in time $\poly{n,m,v_{max}}$, where $v_{max} = \max_{i,j} v_{ij}$. In particular, when $v_{max} \le \poly{m,n}$, an EQ1+fPO allocation can be computed in $\poly{m,n}$ time.
\end{theorem}

Finally using Lemma~\ref{lem:sparse-poly-utils}, Theorem~\ref{thm:eq1po-main} becomes:
\begin{theorem}\label{thm:eq1po-sparse}
Given a $k$-ary fair division instance $I = (N,M,V)$ where $k$ is a constant, an allocation that is EQ1 and fPO can be computed in time $\poly{m,n}$.
\end{theorem}

\begin{remark}\normalfont
We remark that our techniques can be used to show that EQ1+fPO allocations of \textit{chores} can be computed in pseudo-polynomial time and in polynomial-time for $k$-ary instances with constant $k$. In the chores model, agent $i$ incur a cost or disutility $c_{ij} \ge 0$ on being assigned the chore $j$. For chores, an allocation $\x$ is said to be EQ1 if for all agents $i, h$, there exists a chore $j\in\x_i$ s.t. $c_i(\x_i\setminus \{j\}) \le c_h(\x_h)$. Note that while comparing bundles, an agent removes a chore from her own bundle in the case of chores, while an agent removes a good from other's bundle in the case of goods. To obtain an EQ1+fPO allocation of chores, Algorithm~\ref{alg:eq1po} can be modified to perform transfers from an EQ1-violator agent $h$, where $c_h(\x_h\setminus j) > c_i(\x_i)$ for every $j\in \x_h$ and a least disutility agent $i$. The convergence analysis for this modification of the algorithm for chores is similar to that of Algorithm~\ref{alg:eq1po}. We also note that in general, the existence of EF1+PO allocations of chores is open.
\end{remark}

\section{Finding EF1+PO allocations for constant number of agents}\label{sec:const-n}
In this section, we show how to compute an EF1+PO allocation in polynomial time when $n$, the number of agents, is a constant. Our algorithm relies on the fact that there \textit{exists} an EF1+fPO allocation for every instance and thus aims to find such an allocation by effectively enumerating over all fPO allocations. For constant $n$, we argue that this enumeration takes time $\poly{m}$. We borrow some terminology from \cite{branzei2019choresceei} to describe this enumeration technique. 

Given a fractional allocation $\z$, the \textit{consumption graph} $G_{\z}$ is defined to be a bipartite graph $(N,M,E)$, where $(i,j)\in E$ iff $z_{ij} > 0$. Let $P = (i_1, j_1, i_2, j_2, \dots, i_K, j_K, i_{K+1})$, be a path in a consumption graph $(N,M)$, where $K\ge 1$, and $v_{ij} > 0$ for any $(i,j)\in P$. We define the product of utilities along $P$ as: 
\begin{equation}\label{eq:pivalue}
    \pi(P) = \prod_{k=1}^K\frac{v_{i_k j_k}}{v_{i_{k+1} j_k}}.
\end{equation}
When $i_{K+1}=i_1$, $P$ is a cycle. We now characterize fPO allocations based on the properties of their associated consumption graphs.
\begin{lemma}[\cite{branzei2019choresceei}]\label{lem:fpocycle} An allocation $\z$ is fPO iff for every cycle $C$ in its consumption graph, $\pi(C) = 1$.
\end{lemma}
\begin{proof}(sketch)
See Corollary 16 of \cite{branzei2019choresceei} for the full proof. Intuitively, if $\pi(C) > 1$, then transferring small amounts of goods along the cycle will result in an allocation $\z'$, which is a Pareto-improvement over $\z$, contradicting the fact that $\z$ is fPO. Likewise, if $\pi(C) < 1$, performing the transfer in the opposite order results in a Pareto improvement.
\end{proof}

Thus, the existence of cycles $C$ in a consumption graph of an allocation with $\pi(C)$ depends on algebraic properties of the (non-zero) values $v_{ij}$ of the instance. This motivates the definition of \textit{non-degenerate instance} as an instance where the values $v_{ij}$ share no multiplicative relationship. Formally:
\begin{definition}[Non-degenerate instance] A fair division instance $I = (N, M, V)$ is said to be non-degenerate if for every path $P = (i_1, j_1, i_2, j_2, \dots, i_K, j_K, i_{K+1})$, in the complete bipartite graph $(N,M)$, where $K\ge 1$, and $v_{ij} > 0$ for any $(i,j)\in P$, it holds that $\pi(P) \neq 1$.
\end{definition}

For an allocation $\z$, let $\mathbf{u}(\z) \in \R^n$ be the corresponding utility vector, where $\mathbf{u}(\z)_i = v_i(\z_i)$ for each $i\in N$.
\cite{branzei2019choresceei} showed that for each fPO utility vector $\mathbf{u}$ of a non-degenerate instance, there is a unique feasible (fractional) allocation $\z$ such that $\mathbf{u}(\z) = \mathbf{u}$. \cite{branzei2019choresceei} showed how to enumerate fPO allocations using the following definition.
\begin{definition}[Rich family of graphs] A collection of bipartite graphs $\mathcal{G}$ is said to be \textit{rich} for a given instance $(N,M,V)$ if for any fPO utility vector $\mathbf{u}$, there is a feasible allocation $\z$ with $\mathbf{u}(\z) = \mathbf{u}$ such that the consumption graph $G_\z$ is in the collection $\mathcal{G}$.
\end{definition}
Thus, a rich family of graphs contains the consumption graphs of every fPO utility vector for the instance. \cite{branzei2019choresceei} show how to construct a rich family of graphs $\mathcal{G}$ for every instance.
\begin{theorem}[Proposition 23 of \cite{branzei2019choresceei}] For constant number of agents $n$, a rich family of graphs $\mathcal{G}$ can be constructed in time $O(m^{\frac{n(n-1)}{2}+1})$ and has at most $(2m+1)^{\frac{n(n-1)}{2}+1}$ elements.
\end{theorem}

Our algorithm for finding an EF1+fPO allocation in a non-degenerate instance $I$ with a constant number of agents is as follows. We generate a rich family of graphs in $\poly{m}$ time since $n$ is constant. We know there exists an integral EF1+fPO allocation $\x$, for instance $I$, with corresponding utility profile $\mathbf{u}(\x)$. Since $I$ is non-degenerate, $\x$ is the only allocation corresponding to the utility profile $\mathbf{u}(\x)$. Now since the collection $\mathcal{G}$ is rich, for the utility profile $\mathbf{u}(\x)$, the consumption graph $G_\x$ of $\x$ is present in $\mathcal{\G}$. Since $\x$ is integral, the degree of every good in $G_\x$ is one. Thus, we can find $\x$ by iterating over every graph $G$ in $\mathcal{G}$, checking if the degree of every good in $G$ is one, i.e., corresponds to an integral allocation, and then checking if the integral allocation is EF1. This algorithm runs in polynomial time since $\mathcal{G}$ has $\poly{m}$ graphs as $n$ is constant, and checking if an allocation is EF1 can also be done in $\poly{m}$ time.

We now show how to adapt our algorithm for all instances, not just non-degenerate instances. Given a fair division instance $I = (N,M,V)$, we use the perturbation technique of \cite{duan2016CE} to construct a non-degenerate instance $I' = (N,M,V')$ as follows. For each $i\in N$ and $j\in M$, we choose a distinct prime number $q_{ij}$. The prime number theorem implies that for every integer $Q$, there are at least $\frac{Q}{2\ln Q}$ primes less than or equal to $Q$. Thus for $Q = 8nm\ln(nm)$, there are at least $nm$ primes at most $Q$, which can be computed in $O(Q\ln Q)$ time. The values in the perturbed instance $I'$ are then defined as $v'_{ij} = v_{ij}\cdot q_{ij}^\eps$, for a small constant $\eps$ given by $\eps := \log_Q (1+\frac{1}{2mv_{max}}) \in (0,1)$.

We now run our algorithm on the non-degenerate instance $I'$ and compute an EF1+fPO allocation $\x$. We argue that $\x$ is EF1+PO for the original instance $I$ as well.

\begin{lemma}\label{lem:const-n-ef1}
If $\x$ is EF1 for $I'$, then $\x$ is EF1 for $I$.
\end{lemma}
\begin{proof}
Consider any agent $i\in N$ and two sets $S,T\subseteq M$ such that $v_i(S)>v_i(T)$. First note:
\[ v'_i(S) - v_i(S) = \sum_{j\in S} v_{ij}\cdot(q_{ij}^\eps - 1) \ge 0.\]
Next, observe that:
\[ v'_i(T) - v_i(T) = \sum_{j\in T} v_{ij}\cdot(q_{ij}^\eps - 1) \le m\cdot v_{max} \cdot (Q^\eps - 1) = \frac{1}{2},\]
where the last equality used the definition of $\eps$.
Putting the above two inequalities together, we have:
\[v'_i(S) - v'_i(T) \ge v_i(S) - v_i(T) - \frac{1}{2} > 0, \]
where the last inequality holds because $v_i(S) - v_i(T) \ge 1$, since the original valuations are integral.

Suppose that $\x$ is EF1 for the instance $I'$ and not EF1 for the instance $I$. Then there are two agents $i,h$ s.t. for all $j\in\x_h$: $v_i(\x_h\setminus \{j\}) > v_i(\x_i)$. Then the argument above implies that $v'_i(\x_h\setminus \{j\}) > v'_i(\x_i)$ for every $j\in\x_h$. This contradicts our assumption that $\x$ is EF1 for $I'$.
\end{proof}

\begin{lemma}\label{lem:const-n-po}
If $\x$ is fPO for $I'$, then $\x$ is PO for $I$.
\end{lemma}
\begin{proof}

The Second Welfare Theorem~\cite{mas1995microeconomic} states that for any fPO allocation $\x$ of an instance $([n],[m],V)$, there exists a set of non-negative budgets $\{e_i\}_{i\in[n]}$ and a set of prices of the good $\p\ge 0$, such that $(\x,\p)$ is a market equilibrium for the Fisher market instance $([n],[m],V,\{e_i\}_{i\in [n]})$. 

Therefore, since $\x$ is fPO for $I'$, by the Second Welfare Theorem, there exists a set of prices $\p$ s.t. $(\x,\p)$ is a market equilibrium for an associated Fisher market instance. In particular, $(\x,\p)$ is on MBB. For an agent $i$, let $\alpha'_i = \max_{j\in M} v'_{ij}/p_j$, and let $\alpha_i = \max_{j\in M} v_{ij}/p_j$. Since the prices can be scaled, we can assume that for all $j\in M$, $1 \le p_j \le 2$.

For any agent $i$ and good $j$, by the definition of $v'_{ij}$ and $\eps := \log_Q(1+\frac{1}{2mv_{max}})$, we have that:
\[v_{ij} \le v'_{ij} \le v_{ij}\cdot(1+\delta),\]
where $\delta = \frac{1}{2mv_{max}}$. This means that $\alpha'_i \ge \alpha_i$ for every $i\in N$. Further, since $(\x,\p)$ is on MBB for $I'$, for any $i\in N$:
\[
\p(\x_i) = \frac{v'_i(\x_i)}{\alpha'_i} \le \frac{(1+\delta)v_i(\x_i)}{\alpha_i},
\]
which gives:
\begin{equation}\label{eqn:const-PO}
\sum_{i\in N} \frac{v_i(\x_i)}{\alpha_i} \ge \sum_{i\in N}\frac{\p(\x_i)}{1+\delta}  \ge \frac{\p(M)}{1+\delta}
\end{equation}
Suppose $\x$ is not PO for the instance $I$. Then there is some allocation $\y$ that Pareto-dominates $\x$, i.e., for every $i\in N$, $v_i(\y_i) \ge v_i(\x_i)$, and for some $h\in N$, $v_h(\y_h) \ge v_h(\x_h) + 1$, since the valuations $v_{ij}$ are integral. Further since $\alpha_i$ is the maximum bang-per-buck ratio for an agent $i$ at prices $\p$, we have for every $i\neq h$:
$\p(\y_i) \ge \frac{v_i(\y_i)}{\alpha_i} \ge \frac{v_i(\x_i)}{\alpha_i}$
and
$\p(\y_h) \ge \frac{v_h(\y_h)}{\alpha_h} \ge \frac{v_h(\x_h)+1}{\alpha_h}$. This gives:
\begin{equation}\label{eqn:const-PO-2}
\begin{aligned}
\sum_{i\in N} \frac{v_i(\x_i)}{\alpha_i} &= 
\sum_{i\in N\setminus\{h\}}\frac{v_i(\x_i)}{\alpha_i} + \frac{v_h(\x_h)}{\alpha_h} \\
&\le \sum_{i\in N\setminus\{h\}}\p(\y_i) + \p(\y_h) - \frac{1}{\alpha_h} \\
&= \p(M) - 1/\alpha_h . 
\end{aligned}
\end{equation}
Putting~\eqref{eqn:const-PO} and~\eqref{eqn:const-PO-2} together, we get:
\[
\p(M) - 1/\alpha_h\ge\frac{\p(M)}{1+\delta},
\]
which simplifies to:
\[\delta\alpha_h\p(M) \ge 1+\delta .\]
Notice however that $\alpha_h \le v_{max}$, and $\p(M) \le 2m$, and hence:
\[\delta\alpha_h\p(M) \le \delta \cdot v_{max} \cdot 2m = 1,\]
which is a contradiction. Hence $\x$ must be PO.
\end{proof}

Thus, we have shown:
\begin{theorem}\label{thm:const-n-ef1po}
Given a fair division instance $I = (N,M,V)$, where $n$ is constant, an EF1+PO allocation can be computed in $\poly{m}$ time.
\end{theorem}

Finally, we note that the techniques in Lemma~\ref{lem:const-n-ef1} also extend easily to EQ1. Since we showed that an EQ1+fPO allocation is guaranteed to exist for positive instances:
\begin{theorem}\label{thm:const-n-eq1po}
Given a positive fair division instance $I = (N,M,V)$, where $n$ is constant, an EQ1+PO allocation can be computed in $\poly{m}$ time.
\end{theorem}

\section{EF1+PO allocations for $k$-ary instances with constant $n$ and $k$}\label{sec:const-n-k}
We now consider $k$-ary fair division instances $(N,M,V)$ where both $k$ and $n$ (number of agents) are constant.

Let $\mathcal{X}$ be the set of all allocations for the instance $I$. For each agent $i\in N$, let $T_i = \{v_i(\x_i) : \x \in \mathcal{X}\}$, the set of different utility values $i$ can get from any allocation. Let $U = \max_{i\in N} |T_i|$. From Lemma~\ref{lem:sparse-poly-utils}, we know $U$ is at most $\poly{m}$. Define $T = T_1\times\cdots\times T_n$. We note that $|T| \le (\poly{m})^n = \poly{m}$, since $n$ is constant, and can be computed in $\poly{m}$-time.

To solve certain fair division problems for such instances, we enumerate over each entry $(u_1,\dots,u_n)$ of $T$ and check if there is a feasible allocation $\x$ in which each agent $i$ gets utility exactly $u_i$. The next lemma shows that the latter can be done efficiently.

\begin{lemma}\label{lem:const-n-k}
Given a vector $(u_1,\dots,u_n)\in T$, it can be checked in $\poly{m}$-time whether there is a feasible allocation $\x$ s.t. for all agents $i$, $v_i(\x_i) = u_i$.
\end{lemma}
\begin{proof}
For each agent $i\in N$, let $S_i = \{v_{ij} : j\in M\}$ be the set of values $i$ has for the goods. For each good $j\in M$, define $\s{label}(j)$ to be the position of the vector $(v_{1j},\dots,v_{nj})$ when the elements of $S_1\times \cdots \times S_n$ are ordered lexicographically. Let $L$ be the set of labels. Clearly $|L| \le |S_1|\cdots|S_n| = O(1)$ since each $|S_i| \le k$, and $k$ and $n$ are constants. Essentially, this means that there are $|L| = O(1)$ different \textit{types} of goods. For a label $1\le \ell\le |L|$, and for any agent $i$, let $v_{i\ell}$ equal $v_{ij}$ for any good $j$ with $\s{label}(j) = \ell$. Further, let $m_\ell$ be the number of goods with label $\ell$. All goods can be labeled in $\poly{m}$ time.

For each agent $i\in N$, let $T_i = \{v_i(\x_i) : \x \in \mathcal{X}\}$, the set of different utility values $i$ can get from any allocation. Let $U = \max_{i\in N} |T_i|$. From Lemma~\ref{lem:sparse-poly-utils}, we know $U$ is at most $\poly{m}$. Define $T = T_1\times\cdots\times T_n$. We note that $|T| \le (\poly{m})^n = \poly{m}$, since $n$ is constant.

To solve certain fair division problems for such instances, we simply enumerate over each entry $(u_1,\dots,u_n)$ of the set $T$ and check if there is a feasible allocation $\x$ in which each agent $i$ gets utility exactly $u_i$. 

For the latter, we define integer variables $m_{i\ell}$ for each $i\in N$ and $1\le \ell \le |L|$, which represents the number of goods with label $\ell$ assigned to agent $i$. Now consider the following linear system:
\begin{align*}
\forall i\in N:& \quad\sum_{\ell=1}^{|L|} m_{i\ell}v_{i\ell} = u_i \\
\forall \ell \in [|L|]:& \quad\sum_{i\in N} m_{i\ell} = m_\ell \\
\forall i\in N,\, \forall \ell\in [|L|]:& \quad 0\le m_{i\ell} \le m_\ell
\end{align*}
This system has $n|L|$ variables, and each variable $m_{i\ell}$ can take at most $m+1$ values. Therefore, this system has a constant dimension, and each variable is in a bounded range. Thus, by simple enumeration, we can check in $\poly{m}$ time whether this system has a solution or not. If it does, then an allocation $\x$ exists, which gives every agent $i$ a utility of $u_i$.

Therefore, given a vector $(u_1,\dots,u_n)\in T$, checking if there is a feasible allocation $\x$ in which each agent $i$ gets utility exactly $u_i$ can be done in $\poly{m}$-time. 
\end{proof}

By iterating through $T$, we can prepare a list of feasible utility vectors (and corresponding allocations) that satisfy our fairness and efficiency criteria in $\poly{m}$-time.

\begin{theorem}\label{thm:const-n-k}
For $k$-ary instance $I = (N,M,V)$ where both $k$ and the number of agents $n$ are constants, we can compute in $\poly{m}$ time (i) an MNW allocation, (ii) a leximin optimal allocation, (iii) a $\mathcal{F}$+fPO allocation (when it exists) where $\mathcal{F}$ is any polynomial-time checkable property.
\end{theorem}

\section{$\s{PLS}$ membership of finding an EF1+fPO allocation}\label{sec:pls}
In this section, we show that the problem of computing an EF1+fPO allocation lies in the complexity class $\s{PLS}$ \cite{JOHNSON1988PLS}. Essentially, $\s{PLS}$ captures local search problems whose local optimality can be verified in polynomial time. We, therefore, phrase the problem of computing an EF1+fPO allocation in a fair division instance $I = (N,M,V)$ as a local search problem $\Phi$. We closely follow our Algorithm~\ref{alg:ef1po}, which computes an EF1+fPO allocation and shows that it has the structure of a local search problem. We now describe the solution space, cost function, and neighborhood structure of $\Phi$, following which we show that the local maxima of $\Phi$ correspond to EF1+fPO allocations for the instance $I$. 

\textit{Solution space.} We first define a configuration space as follows. Let $\x^0$ be a specific initial integral fPO allocation. Each element of the configuration space is of the form $(\x, L, \varphi)$, where $\x$ is an integral fPO allocation, $L \subseteq N$ is a set of agents, and $\varphi\in\{0,1,\dots,U\}^n$, where $U$ is the maximum utility an agent can get in any allocation. Clearly, the solution space is finite and can be represented in size polynomial in the representation of the instance $I$. Moreover, since it can be checked in polynomial time via a linear program whether a given allocation is fPO \cite{Barman18FFEA}, we can efficiently check if a given tuple $(\x,L,\varphi)$ is a member of the configuration space or not.

Given a configuration $(\x, L, \varphi)$, we use the convex program \eqref{prg:prices} below to find a set of unique prices $\p$ s.t. $\p$ s.t. $\sum_j p_j = n$, $(\x, \p)$ is on MBB and $L$ is the set of least spenders in $(\x, \p)$. The vector $\varphi$ intuitively captures the potential at the allocation $(\x,\p)$ during the run of our EF1+fPO Algorithm~\ref{alg:ef1po} starting with $\x^0$ as the initial allocation. That is, $\varphi_i$ stores the utility of agent $i$ the last time $i$ was the LS in the run of Algorithm~\ref{alg:ef1po} starting from $(\x^0, \p^0)$, where $\q^0$ is the solution to the program \eqref{prg:prices} for the allocation $\x^0$. 

We now describe the convex program \eqref{prg:prices}, which for a given configuration $(\x, L, \varphi)$ finds a set of unique prices $\p$ s.t. $\sum_j p_j = n$, $(\x, \p)$ is on MBB, and $L$ is the set of least spenders in $(\x, \p)$.
\begin{equation}\label{prg:prices}
\begin{aligned}
\text{maximize} & \ \prod_j p_j \\
\forall i\in N, \forall j\in M: & \ \ p_j \ge \lambda_i \cdot v_{ij} \\
\forall i\in N, \forall j\in M, \x_{ij} > 0: & \ \ p_j = \lambda_i \cdot v_{ij} \\
\forall i\in N: & \ \ \p(\x_i) \ge \mu \\
\forall i\in L: & \ \ \p(\x_i) = \mu \\
& \sum_j p_j = n 
\end{aligned}
\end{equation}

The first two constraints capture the MBB condition; here, $\lambda_i$ is a variable that represents the reciprocal of the MBB ratio of agent $i$. The variable $\mu$ captures the spending of the least spenders. Finally, the constraint $\sum_j p_j = n$ and the convexity of the objective $\prod_j p_j$ ensure that the set of prices respecting the MBB and LS constraints are unique and positive. The convex program can be solved in polynomial time since the number of constraints is polynomial in $n$ and $m$.

\textit{Cost function.} The cost of a configuration $(\x, L, \varphi)$ is a lexicographic cost function $\langle \delta(\x), \varphi\rangle$, where $\delta(\x) \in \{-1,0,1\}$. If $\x$ is EF1, then $\delta(\x) = 1$. If $\x$ is not EF1, then $\delta(\x)$ is either $-1$ or $0$, depending on whether the configuration is valid or not. Let $\p$ be the set of prices obtained by solving the program \eqref{prg:prices} for the allocation $\x$. If $L$ does not equal the set of LS in $(\x, \p)$, then $\delta(\x) = -1$. Otherwise if $v_i(\x_i) < \varphi_i$ for a LS $i$ in $(\x,\p)$, then also $\delta(\x) = -1$. These cases correspond to the configuration $(\x, L, \varphi)$ being invalid, i.e., $L$ is not the correct set of LS, or $\varphi_i$ cannot be the spending of $i$ the last time $i$ was an LS in a run of Algorithm~\ref{alg:ef1po} from $(\x^0,\p^0)$. In the other case, $(\x, L, \varphi)$ is valid and if $\x$ is not EF1, then $\delta(\x) = 0$. We note that the cost of configuration can be computed in polynomial time.

\textit{Neighborhood structure.} A configuration $(\x, L, \varphi)$ with $\delta(\x) = 1$ has no neighbor. When $\delta(\x) = -1$, its neighbor is the allocation $(\x^0, L^0, 0^n)$, where $\x^0$ is the initial allocation and the set $L^0$ is the set of LS in $(\x^0, \p^0)$ for the prices $\p^0$ corresponding to $\x^0$.
When $\delta(\x) = 0$, the configuration is valid and $\x$ is not EF1. We then compute $\p$ using \eqref{prg:prices} for the allocation $\x$, and run our Algorithm~\ref{alg:ef1po} with $(\x, \p)$ as the initial configuration until the set of LS changes and we reach an allocation $(\x', \p')$, where the set of LS is $L'$. We then update the potential function to $\varphi'$ whose $i^{th}$ entry records the spending of agent $i$ the last time $i$ was an LS. The configuration $(\x', L', \varphi')$ is the neighbor of the $(\x, L, \varphi)$. From Lemma~\ref{lem:swap-poly}, we know the set of LS changes in $\poly{n,m}$ iterations, so the neighbor of a configuration can be computed in polynomial time. 

\textit{Membership in $\s{PLS}$.} Having defined the solution space, cost function, and the neighborhood structure of the local search problem $\Phi$, we show that it is in $\s{PLS}$. $\s{PLS}$ membership follows from demonstrating the following three polynomial time algorithms, which are implicit in the preceding paragraphs explaining the cost function and the neighborhood structure. 

\begin{enumerate}
\item Algorithm A: Which outputs the initial allocation $(\x^0, \p^0)$.
\item Algorithm B: Which outputs the cost of a configuration.
\item Algorithm C: Which outputs a neighbor with strictly higher cost or is locally optimal.
\end{enumerate}

Therefore, $\Phi$ is in $\s{PLS}$. Lastly, we argue that all the local optima of $\Phi$ correspond to EF1+fPO allocations. This is straightforward, since the local optima of $\Phi$ comprise of configurations $(\x, L, \varphi)$ where $\delta(\x) = 1$, which holds iff $\x$ is EF1. Since $\x$ is fPO from the definition of a configuration, $\x$ is EF1+fPO. We can, therefore, conclude:
\begin{theorem}\label{thm:pls}
The problem of computing an EF1+fPO allocation lies in $\s{PLS}$.
\end{theorem}

The standard algorithm for $\s{PLS}$ problems uses Algorithm A to find an initial solution and repeatedly uses Algorithms B and C to find neighbors with higher cost until a local optimum is reached. This standard algorithm when applied to $\Phi$ results in a sequence of allocations $(\x^0, \x^1, \dots, \x^T)$, where $\x^T$ is EF1+fPO, $T\le \poly{n,m,U}$, and this sequence of allocations is a subsequence of allocations encountered in the run of Algorithm~\ref{alg:ef1po} starting from $(\x^0, \p^0)$.

Lastly, we remark that the above result does not show that the problem of computing an EF1+PO allocation is in $\s{PLS}$. In fact, since it is coNP-complete to check if a given allocation is PO, the problem of computing an EF1+PO allocation is not even in $\s{TFNP}$ (unless $\s{P}=\s{NP}$).

\section{Discussion}
In this paper, we showed that an EF1+fPO allocation can be computed in pseudo-polynomial time, thus improving upon the result of Barman, Krishnamurthy, and Vaish \cite{Barman18FFEA}. Our work establishes the polynomial time computability of EF1+PO allocations for two large non-trivial subclasses of instances: (i) $k$-ary valuations with constant $k$, and (ii) constant $n$ (number of agents). These results are especially significant because polynomial-time computability was previously known only for the simple classes of binary or identical valuations. Moreover, computing the MNW allocation remains NP-hard for these classes, thus eliminating its use for efficient computation of an EF1+PO allocation. Further, these classes could be useful practically when the number of agents is small or the valuations are derived from asking the agents to rate items on a small scale. Our results also extend to the fairness notions of EQ1 and weighted EF1. 

On the complexity front, we showed that computing an EF1+fPO lies in $\s{PLS}$. Settling the complexity of the EF1+fPO problem by designing a polynomial time algorithm or showing $\s{PLS}$-hardness remains a challenging open question. Finally, showing the existence of EF1+PO allocations for general additive instances of chores is another interesting research direction.

\section*{Acknowledgements}
We would like to thank Ruta Mehta for discussion and insights regarding the proof of Theorem~\ref{thm:pls}. Work on this paper is supported by NSF Grant CCF-1942321.

\bibliography{references}

\end{document}